\documentclass[11pt, a4paper]{article}
\usepackage{amsfonts,amsmath,amsthm,amssymb,mathrsfs}
\usepackage{graphicx}	
\newcommand{\ket}[1]{\left| #1 \right>} 
\newcommand{\bra}[1]{\left< #1 \right|} 
 
\usepackage[english]{babel}
\usepackage{color}
\DeclareMathOperator{\e}{e}
\def\assign{\mathrel{:=}}
\def\diag{\mathop{\rm diag}}

\theoremstyle{proposition}
\newtheorem{proposition}{Proposition}

\title{A Quantum Key Distribution protocol for qu$d$its with better noise resistance}
\author{
  Zo\'e AMBLARD\\
  XLIM Laboratory\\
  University of Limoges
  \and
  Fran\c{c}ois ARNAULT\\
  XLIM Laboratory\\
  University of Limoges
}
\begin{document}
\maketitle

\bigskip

\bigskip

\bigskip

\begin{abstract}
The Ekert quantum key distribution protocol~\cite{Ekert1991} uses pairs of entangled qubits and
performs checks based on a Bell inequality to detect eavesdropping.
The N-DEB protocol~\cite{dDEB} uses instead pairs of entangled qu$d$its to achieve better noise
resistance than the Ekert protocol. It performs checks based on a Bell inequality for qu$d$its found in~\cite{CGLMP} and which we will refer to as the CGLMP-$d$.  In this paper, we present the generalization of our protocol h3DEB~\cite{h3DEB} for qu$d$its.  This protocol also uses pairs of entangled qu$d$its, but achieves even better noise resistance than N-DEB and is showed to be secure against the same family of cloning attacks than N-DEB.  This gain of performance is obtained by using another inequality called here hCHSH-$d$, which was discovered in~\cite{Arnault2012}. For each party, the hCHSH-$d$ inequality involves $2d$ observables. We explain how the parties can measure these observables and thus are able to check the violation of hCHSH-$d$. In the presence of noise, this violation allows the parties to ensure the secrecy of the key because it guarantees the absence of a local Trojan horse attack. The advantage of our proposed scheme is that it results in an increased resistance to noise while remaining secure against individual attacks.
\end{abstract}
\newpage

\section{Introduction}

The Ekert91 protocol~\cite{Ekert1991} exploits pairs of entangled states to exchange keys, and uses Bell inequalities to detect eavesdropping.  Some of the measurement results obtained by the two parties Alice and Bob are perfectly correlated, providing key bits.  Other measurement results must exhibit quantum behaviour if there is no alteration of the quantum channel, and this permits to detect eavesdropping by testing a Bell inequality violation.

The amount of quantum violation is an important characteristic in key distribution protocols because a larger violation is one of the factors that lead to a better noise resistance~\cite{ViolAndNoise}.   Some progress has been made to increase this amount of violation with the use of parties with higher dimension~\cite{ViolStrongThan} or specific entangled states. One can also consider choosing different Bell inequalities to detect eavesdropping.

In their article introducing the N-DEB protocol~\cite{dDEB}, Durt, Kaszlikowski, Chen and Kwek use $d$-dimensional quantum systems (qu$d$its), and the CGLMP-$d$ inequality to obtain better noise resistance than for the Ekert'91 protocol.

Our work makes one step further by using a recent discovered Bell inequality (here called hCHSH-$d$), which belongs to the family of homogeneous Bell inequalities introduced in~\cite{Arnault2012}.  The amount of violation which can be achieved with entangled states is even better than for the CGLMP-$d$. Consequently, the protocol we derive is more tolerant to noise than N-DEB.

Devices called multiport beam splitters~\cite{Tritters} (or ditters), are mentioned in~\cite{dDEB} as one way to handle measurements of qu$d$its.  Ditters are analyzed in~\cite{tritter1} and experimentally tested for $d=3,4$ in~\cite{tritter2}.  Our new protocol h$d$DEB described in this article is analysed in view of the use of ditters to implement measurements.  A crucial point here will be that some products of observables, each implemented by a ditter coupled with a measurement in the computational basis, can also be implemented by another single ditter.  This is needed for our protocol as the inequality hCHSH-$d$ involves such products. 

The paper is organized as follows.
It begins with some reminders and precisions about measurements with ditters in Section~2, where we also consider the use of ditters for implementing the product of observables.  Then Section~3 recalls the N-DEB protocol.   After that, Section~4 introduces the Bell inequality hCHSH-$d$ we use and defines our new protocol h$d$DEB.  In Section~5, we study the security of our protocol against individual attacks and we show that our protocol h$d$DEB reaches a compromise between resistance to noise and security. Finally, the paper concludes about the advantage of h$d$DEB providing better resistance to noise.
\newpage

\section{Prerequisites}

In what follows, all the sums will be taken modulo $d$. Our protocol use qu$d$its and observables with $d$ outcomes which we label for readability $1,\omega,...,\omega^{d-1}$ where $\omega$ is the $d^\text{th}$
root of unity ${\omega=e^\frac{2i\pi}{d}}$.  The observables used by the two parties Alice and Bob
will be denoted respectively by $A_i$ and~$B_j$ for some indexes $i$ and~$j$.  We will also use the
correlation functions introduced in~\cite{tritter2}~:
$$
  E(A_iB_j) = \sum_{a,b=1,\omega,...,\omega^{d-1}}  P(A_i=a,B_j=b) \, ab.
$$

\subsection{Measurements with ditters}

A ditter is parameterized by a $d$-uplet $(\varphi_0,\varphi_1, ..., \varphi_{d-1})$ of phase shifts.  For
readability we put ${\theta_j=\exp(i\varphi_j)}$ (for $j=0,1, ... , d-1$) and
${\Theta=(\theta_0, ... , \theta_{d-1})}$.  The ditter performs over a qu$d$it the following unitary
transformation : 
$$
  U_{\Theta} \assign HD_\Theta
  =
  \frac{1}{\sqrt{d}}\sum\limits_{k,l=0}^{d-1}\omega^{kl}\theta_{l}\ket{k}\bra{l}
$$
where the matrices $H$ and $D_\Theta$ are $H=(\omega^{kl})_{0\leq k,l\leq d-1}$ and
$D_\Theta=\diag(\theta_0, ... , \theta_{d-1})$.

  After the transformation performed by the ditter, a measurement in the computational basis is made using $d$ detectors.
This measurement is represented by the observable
$$
  Z = \sum_{k=0}^{d-1} \omega^{k}\ket{k}\bra{k}.
$$
As we assumed the $d$ possible outcomes to be labeled by complex roots of unity, we
use unitary observables.  Thus, the measurement obtained by the combination of the ditter and
the detectors corresponds to the following observable 

\begin{equation}
\begin{split}
Z_{\Theta} \assign D_{\Theta^{*}}H^{\dagger}ZHD_{\Theta}= \sum_{k=0}^{d-1} \theta_{k}\theta_{k+1}^{*}\ket{k+1}\bra{k}
\end{split}
\label{eq:ZTheta}
\end{equation}

which gives us, in the particular case where $\theta_{j}=\theta^{j}$:

\begin{equation}
\begin{split}
Z_{\Theta}= \theta^{d-1}\ket{0}\bra{d-1}  + \sum_{k=0}^{d-2} \theta^{*}\ket{k+1}\bra{k}.
\end{split}
\label{ref1}
\end{equation}

\subsection{Products of incompatible observables}

 Suppose that we have two measurement devices (each one represented by a ditter and a measurement in the computational basis), which implement the observables $Z_\Theta$ and~$Z_\Lambda$ described by Equation~(\ref{eq:ZTheta}), with $\Theta=(\theta_0,\theta_1, ... , \theta_{d-1})$ and  $\Lambda=(\lambda_0,\lambda_1, ... , \lambda_{d-1})$.  Then we need to implement the product observable $Z_\Theta^{i} Z_\Lambda^{j}$ for $i =1, ... , d-2 $ and $j=d-i-1$.

\begin{proposition}
Let define the $d$-uplet of phase shifts $\Gamma=(\gamma_0,\gamma_1, ... , \gamma_{d-1})$. 

For any ${i =1, ... , d-2}$ and $j=d-i-1$, the observable $Z_\Theta^{i} Z_\Lambda^{j}$ verifies :

$$
Z_\Theta^{i} Z_\Lambda^{j} = Z_{\Gamma}^\dagger = \sum_{k=0}^{d-1} \gamma_{k+1}\gamma_{k}^{*}\ket{k}\bra{k+1}
$$
with 
$$
\forall k = 0, ..., d-1  \phantom{ iiiiii }   \gamma_{k} = \theta_{k}\theta_{k+1} ... \theta_{k-i-1}\lambda_{k-i}\lambda_{k-i+1} ... \lambda_{k}.
$$
\end{proposition}

\begin{proof}

From Equation~(\ref{eq:ZTheta}), we write :

\begin{equation*}
\begin{split}
Z_\Theta^{i} Z_\Lambda^{j} & = (\sum_{k=0}^{d-1} \theta_{k}\theta_{k+i}^{*}\ket{k+i}\bra{k}) \times (\sum_{l=0}^{d-1} \lambda_{l}\lambda_{l+j}^{*}\ket{l+j}\bra{l}) \\
& = \sum_{l=0}^{d-1} \theta_{l+j}\theta_{l+i+j}^{*} \lambda_{l}\lambda_{l+j}^{*}\ket{l+i+j}\bra{l} \\
& = \sum_{k=0}^{d-1} \theta_{k-i}\theta_{k}^{*} \lambda_{k+1}\lambda_{k-i}^{*}\ket{k}\bra{k+1}.
\end{split}
\end{equation*}

The generalized Pauli matrix $Z$ in dimension $d$ verifies $Z^{d}=I_{d}$. The matrix $Z$ being unitary, it also verifies $I_{d} = ZZ^{\dagger}$, which gives $Z^{d-1}=Z^\dagger$.

For any observable $Z_{\Omega}$ we have ${Z_{\Omega}^{d-1} = D_{\Omega}^{*}H^{\dagger}Z^{\dagger}HD_{\Omega} = Z_{\Omega}^\dagger}$. In order to rewrite a product observable as a new ditter measurement
 of the form $Z_{\Gamma}^\dagger = \sum_{k=0}^{d-1} \gamma_{k+1}\gamma_{k}^{*}\ket{k}\bra{k+1}$, we need $\gamma_{k+1}\gamma_{k}^{*} = \theta_{k-i}\theta_{k}^{*} \lambda_{k+1}\lambda_{k-i}^{*}$. One of the possible solutions is :
$$
\forall k = 0, ..., d-1  \phantom{ iiiiii }   \gamma_{k} = \theta_{k}\theta_{k+1} ... \theta_{k-i-1}\lambda_{k-i}\lambda_{k-i+1} ... \lambda_{k}.
$$
\end{proof}

From Proposition 1, we conclude that any product observable $Z_\Theta^{i} Z_\Lambda^{j}$ is implementable by a ditter and a detector, with the detector performing a measurement corresponding to the observable $Z^{\dagger}$ instead of $Z$. 

\section{The N-DEB protocol}

We will recall the N-DEB protocol introduced in~\cite{dDEB}.

\subsection{The $d$-dimensional inequality used in N-DEB}

For a given value of $d$, the N-DEB protocol uses the $d$-dimensional inequality introduced in~\cite{CGLMP} which is referred in N-DEB as the generalized CHSH and which we will call in our paper the CGLMP-$d$ inequality. The maximally entangled state
\begin{equation}
\begin{split}
\ket{\psi}=\frac{1}{\sqrt{d}}\sum_{j=0}^{d-1}\ket{jj}
\label{GHZstate}
\end{split}
\end{equation}
is known to violate this inequality with the four bases used in N-DEB. For these bases, we will use the same denomination "optimal bases'' as in~\cite{dDEB}. The violation values for $d=3, 4, 5$ are summarized in the following table :

\begin{table}[h]
\centerline{\begin{minipage}[t]{0.5\linewidth}
\centering
\begin{tabular}{c | c c c c c}
\hline\hline
 $d$ & $v_{d}$ & $N_{d}$\\
\hline
$3$ & $1.436$ & $0.304$\\ 
$4$ & $1.448$ & $0.309$\\ 
$5$ & $1.455$ & $0.313$\\ 
\end{tabular}
\label{tab:hresult1}
\caption{\label{pic:violationsCGLMPd} Violation value of the CGLMP-$d$ inequality with the maximally entangled state for $d = 3, 4, 5$}. 
\end{minipage}}
\end{table}

\subsection{The N-DEB procedure}

Alice uses four observables $A_a$ with $a=0$ to 3, corresponding to ditter measurements
with phase shift $(1, \theta^a, \theta^{2a}, ... , \theta^{(d-1)a})$.  Bob use four observables $B_b$ with $b=0$ to 3, corresponding to ditter measurements with phase shift $(1,\theta^{-b},\theta^{-2b}, ... , \theta^{-(d-1)b})$.
The following steps are repeated until Alice and Bob obtained a shared key of desired length.

\begin{enumerate}
\item Alice and Bob obtain an entangled pair of states in the maximally entangled state defined in~(\ref{GHZstate}).

\item Alice draws randomly a value for $a \in \{0,1,2,3\}$ and makes the measurement corresponding to the observable $A_{a}$ whereas Bob draws randomly a value for $b\in \{0,1,2,3\}$ and makes the measurement corresponding to the observable $B_{b}$. 

\item When $a=b$, the results obtained by Alice and Bob are perfectly correlated. Indeed, the two
ditters used by Alice and Bob perform on the shared maximally entangled state the transformation ${(H \otimes H)(D_{\Theta} \otimes D_{\Theta^{*}})}$ with $\Theta=(1,\theta^a,\theta^{2a}, ... , \theta^{(d-1)a})$.  But it is easy to check that:
\begin{equation}
{(H \otimes H)(D_{\Theta} \otimes D_{\Theta^{*}})(\frac{1}{\sqrt{d}}\sum\limits_{j=0}^{d-1}\ket{jj})=\frac{1}{\sqrt{d}} \sum\limits_{\substack{k,k'=0 \\ k+k'\equiv 0 [d] }}^{d-1}\ket{kk'}}.
\label{CollectTrits}
\end{equation}
Consequently, in this case where $a=b$, Alice and Bob obtain a new dit for the shared key.

\item When $a \neq b$, Alice and Bob can use a fraction of their joint measurements to detect eavesdropping by checking a configuration of maximal violation of the CGLMP-$d$.
\end{enumerate}

\section{The hdDEB protocol}
 
We will now describe our protocol.  It achieves better noise resistance because it uses an
homogeneous Bell inequality, which has a larger violation factor than CGLMP-$d$.

\subsection{Violation of the inequality hCHSH-$d$}

Depending on the entangled state that we want to use in our protocol, we can choose an inequality belonging to the set of homogeneous Bell inequalities described in~\cite{Arnault2012} and which will be called hCHSH-$d$.
It has also been shown in~\cite{Arnault2012} that the homogeneous Bell inequalities are satisfied under the hypothesis of local realism, and that they form a complete set.  An homogeneous Bell inequality for two parties can in general be written

\begin{equation}
\text{Re}(\frac{\rho}{{d^2}\text{cos}(\frac{\pi}{d})}E(T)) \leqslant 1
\label{InegHomog}
\end{equation}

where $\rho = \e^\frac{i\pi}{d}$ and $T$ is an homogeneous polynomial in some measurements Alice and Bob can make.  We call $T$
a {\sl Bell operator}.

A feature of the homogeneous Bell inequalities is that~$T$ involve some products of observables (for example $A_1^3A_2$, $A_1^2A_2^2$ and $A_1A_2^3$ for Alice in the case ${d=5}$) which become incompatible when considered as quantum observables.  The outcomes of such a product of course cannot be meant to be the products of outcomes of incompatible observables. In Proposition 1, we show that if we use the unitary observables $Z_\Theta$ defined in~(\ref{eq:ZTheta}) for the $A_i$ previously described, the product is also a unitary observable which outcomes can be obtained with a single measurement. We also conclude that we can perform this product measurement in terms of a new ditter operation and a final detection in the computational basis.

The local realistic elements $A_1^{d-1}$, $A_1^{i}A_2^{j}$ (for $i =1, ... , d-2 $ and $j=d-i-1$) and
$A_2^{d-1}$ for Alice have to be replaced by the observables
$$
  Z_{\Theta_A}^\dagger, \quad Z^\dagger_{\Gamma_{ij_{A}}},\quad  Z_{\Lambda_A}^\dagger
$$
where the $Z^\dagger_{\Gamma_{ij_{A}}}$ is a product observable as described in Proposition 1 and $\Theta_A$, $\Lambda_A$ are the parameters corresponding to the optimal bases.  Similarly, the party Bob has to use the observables $Z_{\Theta_B}^\dagger$, $Z^\dagger_{\Gamma_{ij_{B}}}$, $Z_{\Lambda_B}^\dagger$.

\medskip

  After substituting these observables to the variables in a Bell operator~$T$, a quantum state $\ket{\psi}$ violates the homogeneous Bell inequality associated to~$T$
with a violation factor $v \geqslant 1 $ if it verifies (compare to~(\ref{InegHomog})) :

\begin{equation}
\frac{1}{{d^2}\text{cos}(\frac{\pi}{d})}\text{Re}(\bra{\psi}\rho T\ket{\psi}) = v \text{  }\text{ for } \rho = \e^\frac{i\pi}{d}.
\label{BellOpe}
\end{equation}

\subsection{The hdDEB procedure}

As for the NDEB protocol described in~\cite{dDEB}, we denote $\mathscr{A}_a = A_0^{d-1-a}A_1^{a}$ with $a=0,1,2,..,d-1$ the observable parameterized by phase shift $(1,\theta^a,..,\theta^{(d-1)a})$, and $\mathscr{B}_b = B_0^{d-1-b}B_1^{b}$ with ${b=0,1,2,..,d-1}$ the observable parameterized by $(1,\theta^{-b},\theta^{-(d-1)b})$. Each of these observables is expected to be implemented with a single ditter from Proposition 1.

\begin{enumerate}
\item Alice and Bob obtain an entangled pair of states in the $d$-dimensional entangled state $\ket{\psi} \assign \sum\limits_{j=0}^{d-1}\delta_{j}\ket{jj}$ with $\delta_{j} \in \mathbb{C}$ for some $j = 0, ... , d-1$.
\item Alice draws randomly a value of $a$ and performs her measurement in the basis associated to the observable $\mathscr{A}_{a}$ whereas Bob draws randomly a value of $\mathscr{B}_b$ and performs his measurement in the basis associated to $B_{b}$.
\item When $a=b$, the results obtained by Alice and Bob are perfectly correlated and they obtain a new dit for the shared key. For completeness, a proof of this statement is given in Appendix A.
\item For some choices of $a$ and $b$ Alice and Bob collect the issues
of their measurements in order to detect eavesdropping by checking a violation of the homogeneous Bell inequality hCHSH-$d$ considered.
\end{enumerate}

\subsection{Choice of the inequality and resistance to noise for ${d = 3, 4, 5 }$}

With the four "optimal bases'' described in~\cite{dDEB}, the maximally entangled states don't allow to reach the largest violations. We consider several non-maximally entangled states which, when used with their corresponding homogeneous Bell inequalities, reach largest violations. A precise derivation of the case $d=3$ can be found in~\cite{h3DEB}.

\medskip

We compare the violations for the following states $\ket{\psi_{d}}$ with $d = 3, 4, 5$  :

\begin{equation*}
\ket{\psi_{3}} = \frac 1{\sqrt{3}} (\ket{00}+\ket{11}+\ket{22}).
\end{equation*}

\begin{equation*}
\ket{\psi_{4}} = \frac 12(\ket{00}+\ket{11}+\ket{22}+\ket{33}).
\end{equation*}

\begin{equation*}
\ket{\psi_{5}} = \frac 1{\sqrt{5}} (\ket{00}+\ket{11}+\ket{22}+\ket{33}-i\ket{44}).
\end{equation*}

We use the homogeneous inequalities associated to each $\ket{\psi_{d}}$ :

\begin{equation}
\frac{1}{{d^2}\text{cos}(\frac{\pi}{d})} \text{Re}(\rho E(T_{d})) \leqslant 1
\label{HBI}
\end{equation}

with 

\medskip

$
 T_{3}
  = - \bigl[(\omega-4)(A_1^2B_1^2) + (\omega+2)(A_1^2B_1B_2) + (\omega-1)(A_1^2B_2^2) \\
\phantom{espace} + (\omega+5)(A_1A_2B_1^2) + (\omega+2)(A_1A_2B_1B_2) + (\omega+1)(A_1A_2B_2^2) \\
\phantom{espace} + (\omega+5)(A_2^2B_1^2) + (\omega+2)(A_2^2B_1B_2) + (\omega-1)(A_2^2B_2^2)\bigr]. \\
$

\medskip

$
 T_{4}
  = -(3\omega + 1) (A_1^3 B_1^3) - (\omega + 1)(A_1^3 B_1^2 B_2) - 5(\omega - 1) (A_1^3 B_1 B_2^2)  \\
  \phantom{ T_{1} = } - (3\omega - 1) (A_1^3 B_2^3) + (\omega + 1) (A_1^2 A_2 B_1^3) - (\omega + 3) (A_1^2 A_2 B_1^2 B_2) \\
  \phantom{ T_{1} = } - (\omega + 1) (A_1^2 A_2 B_1 B_2^2) - (3\omega + 1) (A_1^2 A_2 B_2^3) + (3\omega + 1) (A_1 A_2^2 B_1^3) \\
  \phantom{ T_{1} = } + (5\omega + 1) (A_1 A_2^2 B_1^2 B_2) - (7\omega + 1)(A_1 A_2^2 B_1 B_2^2) + 3(\omega + 1) (A_1 A_2^2 B_2^3) \\
  \phantom{ T_{1} = } - 5(\omega + 1) (A_2^3 B_1^3) + (\omega - 1) (A_2^3 B_1^2 B_2) + (\omega + 1) (A_2^3 B_1 B_2^2) \\
  \phantom{ T_{1} = } - (\omega - 1) (A_2^3 B_2^3).
    $

\medskip

For conveniency, the Bell operator $T_{5}$ is derived in Appendix B. 
 
\medskip

These states and the Bell operators corresponding to their inequalities were chosen because they reach the best compromise between noise resistance and security against individual attacks, as it will be explained in Section 5. 

\medskip

We summarize each choice of entangled state and Bell operator in the following table  :

\begin{table}[h]
\centerline{\begin{minipage}[t]{0.5\linewidth}
\centering
\begin{tabular}{c | c c c c c}
\hline\hline
 $\ket{\psi_{d}}$ & $T_{d}$ & $v_{d}$ & $N_{d}$\\
\hline
$\ket{\psi_{3}}$ & $T_3$ & $1.505$ & $0.336$\\ 
$\ket{\psi_{4}}$ & $T_4$ & $1.546$ & $0.353$\\ 
$\ket{\psi_{5}}$ & $T_5$ & $1.574$ & $0.365$\\ 
\end{tabular}
\label{tab:hresult1}
\caption{\label{pic:violationshCHSHd} Violation value of the hCHSH-$d$ inequality depending on the entangled state and its Bell operator} 
\end{minipage}}
\end{table}

\section{An alternative version of hdDEB secure against individual attacks}

\subsection{An optimal cloning-based attack for the N-DEB and hdDEB protocols}

A cloning-based attack uses a cloning machine (also known as cloner) to copy an input state. Because of the no-cloning theorem, the clonage is imperfect and the adversary aims to design an optimal cloner which copies a specific set of states as accurately as possible. Depending on the properties of the input state or the family of the cloner, the state can be reproduced with a certain amount of fidelity $F_A$ defined in~\cite{dDEB} by
\begin{equation}
F_A = \bra\psi \rho \ket\psi
\label{fidelity}
\end{equation}
where $\psi$ is the initial pure state, and $\rho$ the density of the clone (not necessarily pure).

\bigskip

In~\cite{dDEB} is described a cloning-based attack which uses a phase-covariant qu$d$it cloning machine. This cloner acts with the same accuracy on each states of the optimal bases used in the N-DEB protocol. All these states are copied with the same fidelity $F_A$ depending on the value of $d$ :

\begin{table}[h]
\begin{minipage}[t]{0.5\linewidth}
\centering
\begin{tabular}{| c | c | c | c | c | c | c | c | c | c}
\hline

\hline
$d$ & $3$ & $4$ & $5$ & $6$ & $7$ & $8$ & $9$ & $\infty$  \\
\hline
$F_A$ & $0.7753$ & $0.7342$ & $0.7080$ & $0.6898$ & $0.6762$ & $0.6657$ & $0.6573$ & $0.5$  \\
\hline
\end{tabular}
\end{minipage}
\end{table}

\begin{proposition}
The $2d$ bases considered in our protocol are copied with maximal fidelity when using the optimal phase-covariant cloner described in~\cite{dDEB}.
\end{proposition}

\begin{proof}
Four of our bases are the same optimal bases than from~\cite{dDEB}. The $2(d-2)$ remaining bases have vectors of the form :

\begin{center}

$\frac{1}{\sqrt{d}}\sum\limits_{j=0}^{d-1}e^{i\gamma_{j}}\ket{j}$.

\end{center}

But the cloner described in~\cite{dDEB} is optimal for all states of the form :

\begin{center}
$\sum\limits_{j=0}^{d-1}\delta_{j}\ket{jj}$ for all $\delta_{j}$ verifying $ |\delta_{j}|^2 = \frac{1}{d}$.
\end{center}

Consequently, this cloner is also the optimal asymmetric qu$d$it cloner when considering our $2d$ bases.
\end{proof}

\subsection{The violation of an homogeneous Bell inequality as a sufficient condition for security}

 The violation factor is considered very important for the security of the key distribution protocol.  The presence of noise is usually modeled by the replacement of the initial entangled state by
a mixture

\begin{equation}
\begin{split}
  N \frac Id + (1-N) \ket\psi\bra\psi
  \label{Noisy}
\end{split}
\end{equation}

where $N$ is the proportion of noise.  The point is that the presence of noise decreases the experienced violation to $(1-N)v$ and that the protocol is considered useless when the initial state entanglement cannot be detected anymore.  With this criterion, it has been shown that a protocol is resistant to the presence of noise up to a threshold :

\begin{equation}
\begin{split}
N=1-1/v.
 \label{PresenceNoise}
\end{split}
\end{equation}

When using a noisy channel described by~(\ref{Noisy}), the fidelity (as defined in~(\ref{fidelity})) between the input state~$\ket\psi$ and the output state is given by
\begin{equation*}
F_N = \bra{\psi}\rho'\ket{\psi} = -\frac{d-1}{d}N + 1
\end{equation*}
where $\rho'=  N\frac Id + (1-N)\ket\psi\bra\psi$.
The presence of noise $N$ does not erase the non-classicality of the correlations as long as it stays below the value given by~(\ref{PresenceNoise}).  Hence, it is possible to use a channel for secure key distribution if its fidelity satisfies
\begin{equation}
F_N > \frac{d-1}{dv} + \frac{1}{d}.
\label{FidelityChannel}
\end{equation}

\bigskip

Suppose that an adversary Eve uses an optimal cloner which copy an input state with fidelity $F_A$ and introduces a minimal amount of error $1-F_A$ indistinguishable from an unbiased noise. 

\medskip

Eve's attacks won't be detected as long as $F_A \geqslant F_N$.  Hence, the security of the protocol
against individual cloning attacks is guaranteed if we have
$ {\frac{d-1}{dv} + \frac{1}{d} > F_A}$.  This is equivalent to
$v<\frac{d-1}{dF_A - 1}$.

\newpage

By replacing $F_A$ for each $d = 3, 4, 5$, we obtain the conditions :

\begin{table}[h]
\begin{center}
\begin{minipage}[t]{0.5\linewidth}
\centering
\begin{tabular}{| c | c | c}
\hline
   $d$ & Security criterion \\
\hline
$3$ & $v<1.508$\\ 
$4$ & $v<1.549$\\
$5$ & $v<1.575$\\
$6$ & $v<1.593$\\
$7$ & $v<1.607$\\
$8$ & $v<1.618$\\
$9$ & $v<1.627$\\
$\infty$ & $v<2$\\
\hline
\end{tabular}
\end{minipage}
\end{center}
\end{table}

\subsection{Comparison between N-DEB and h$d$DEB under the same security criterion}

By looking at Table 1 and Table 2, we notice that for each $d=3,4,5$ the violation values of CGLMP-$d$ and hCHSH-$d$ are below the security criterion, which ensure the security of the N-DEB and hNDEB protocols against this family of optimal cloning attacks. But it is also noticeable that, for each $d = 3, 4, 5$, there is a wide gap between the violation values attainable by CGLMP-$d$ and the maximal value tolerated by the security threshold. This gap can be closed by the use of our inequality CGLMP-$d$ which reaches largest violation values :

\begin{table}[h]
\begin{center}
\begin{minipage}[t]{0.5\linewidth}
\centering
\begin{tabular}{| c | c | c | c | c |}
\hline
   $d$ & $v_{\text{N-DEB}}$ & $v_{\text{h}d\text{DEB}}$ & Security criterion \\
\hline
$3$ & $1.436$ & $1.505$ & $v<1.508$\\ 
$4$ & $1.448$ & $1.546$ & $v<1.549$\\
$5$ & $1.455$ & $1.574$ & $v<1.575$\\
\hline
\end{tabular}
\end{minipage}
\end{center}
\end{table}

Moreover, this exploitable gap between $v_{\text{N-DEB}}$ and the security criterion increases with the dimension $d$. From this we conclude not only that our protocol tolerates a higher error rate in the channel than N-DEB while remaining secure against the same family of attacks, but also that this amelioration grows for a higher $d$.

\section{Conclusion}
 
Our goal was to generalize our protocol h3DEB~\cite{h3DEB} in any dimension $d$.  By using the homogeneous Bell inequality hCHSH-$d$ which reaches a better violation factor than the CGLMP-$d$ in dimension $d = 3, 4, 5$, our new protocol h$d$DEB obtain a better threshold of noise resistance than N-DEB.

As the inequality hCHSH-$d$ involves products of observables which become incompatible for quantum states, an important fact is the possibility to implement with slightly modified ditters the single observable corresponding to these products.  We showed in $4.2$ that all the observables needed to compute the violation of an homogeneous Bell inequality, including these products of observables, can be implemented by replacing the final measurement with observable $Z$ by a measurement with observable $Z^\dagger$.  Physically, this replacement corresponds just to a permutation of the detectors. 

  The gain in noise resistance of our protocol over N-DEB is due to the use of the inequality hCHSH-$d$.This inequality detects violations of local realism when some measurements are multiplicatively
related.  By using ditters measurements which respect this multiplicative constraints, the parties
running the protocol are able to exploit its larger violation capabilities. 

The use of $2d$ bases instead of four has the drawback of decreasing the effective dit transfer rate (the probability to obtain a key dit decreases from $\frac{1}{4}$ to $\frac{1}{2d}$) and it makes our protocol more complex ($2(d-2)$ supplementary devices), which can be a potential source of added noise. In the other hand, our protocol tolerates a higher threshold of noise than the one in the N-DEB protocol.

In the paper~\cite{dDEB}, the security of the protocol N-DEB against the optimal individual attack was investigated and it was possible to conclude that a violation of the CHSH-$d$ inequality was a sufficient condition to guarantee the security against individual attacks. We study here the security of our protocol h$d$DEB against the same optimal individual attack and we conclude that our protocol is also secure against this cloning attack. 

For the same level of security against individual attacks, we consequently obtain a better noise resistance than N-DEB and this amelioration is more and more visible when $d$ increases.

\bigskip

\noindent{\bf Acknowledgement}.  One of the author (Z.A.) was partially supported by Thales Alenia
Space during this work.

\bibliographystyle{unsrt}
\bibliography{biblio}

\section*{Appendix A}

We show that in the h$d$DEB procedure, when $a=b$, the results obtained by Alice and Bob are perfectly correlated. 

\medskip

Let first define the $d$-dimensional entangled state $\ket{\psi} \assign \sum\limits_{j=0}^{d-1}\delta_{j}\ket{jj}$ with $\delta_{j} \in \mathbb{C}$ for some $j = 0, ... , d-1$.

\medskip

The two ditters used by Alice and Bob perform on the state $\ket{\psi}$ the transformation ${(H \otimes H)(D_{\Theta} \otimes D_{\Theta^{*}})}$, with ${\Theta=(1,\theta^a,\theta^{2a}, ... , \theta^{(d-1)a})}$. We use the notations :

$$
HD_\Theta
  \assign
  \frac{1}{\sqrt{d}}\sum\limits_{k,l=0}^{d-1}\omega^{kl}\theta_{l}\ket{k}\bra{l}
$$
where the matrices $H$ and $D_\Theta$ are $H=(\omega^{kl})_{0\leq k,l\leq d-1}$ and
$D_\Theta=\diag(\theta_0,\theta_1, ... , \theta_{d})$.

We write :

\begin{equation*}
\begin{split}
(H \otimes H)(D_{\Theta} \otimes D_{\Theta^{*}}) & = (HD_\Theta) \otimes (HD_{\Theta^{*}}) \\
& = \frac1d (\sum\limits_{k,l=0}^{d-1}\omega^{kl}\theta_{l}\ket{k}\bra{l} \otimes \sum\limits_{k',l'=0}^{d-1}\omega^{k'l'}\theta_{l'}\ket{k'}\bra{l'}) \\
& = \frac1d (\sum\limits_{k,l, k', l'=0}^{d-1}\omega^{kl+k'l'}\theta_{l}\theta_{l'}\ket{kk'}\bra{ll'})
\end{split}
\end{equation*}

By applying this transformation to the state $\ket{\psi}$, we obtain :

\begin{equation*}
\begin{split}
(HD_\Theta \otimes HD_{\Theta^{*}})(\sum\limits_{j=0}^{d-1}\delta_{j}\ket{jj}) &= \frac{1}{d} (\sum\limits_{j=0}^{d-1}\sum\limits_{k,k'=0}^{d-1}\delta_{j}\omega^{j(k+k')}\ket{kk'})
\end{split}
\end{equation*}

From $1 + \omega + \omega^2 + ... + \omega^{d-1} = 0$ we finally find :

\begin{equation*}
\begin{split}
(HD_\Theta \otimes HD_{\Theta^{*}})(\sum\limits_{j=0}^{d-1}\delta_{j}\ket{jj}) &= (\sum\limits_{j=0}^{d-1}\delta_{j}) (\sum\limits_{\substack{k,k'=0 \\ k+k'\equiv 0 [d]}}^{d-1}\ket{kk'})
\end{split}
\end{equation*}

\section*{Appendix B}

In order to obtain the violation $ v \simeq 1.574$, we use the entangled state 

\begin{equation*}
\ket{\psi_{5}} = \frac 1{\sqrt{5}} (\ket{00}+\ket{11}+\ket{22}+\ket{33}-i\ket{44}).
\end{equation*}

with the following Bell operator :

\medskip

   $
 T_{5}
  = 
(2\omega^3 - 3\omega + 6) (A_1^4 B_1^4) - (4\omega^2 + 6\omega + 5) (A_1^4 B_1^3 B_2) \\
   \phantom{ T_{1} = } + (7\omega^3 - \omega^2 + 2\omega - 3) (A_1^4 B_1^2 B_2^2) + 
    (\omega^3 + \omega^2 - 4\omega - 3) (A_1^4 B_1 B_2^3) \\
    \phantom{ T_{1} = } + 2(\omega^3 + \omega^2 + 3\omega) (A_1^4 B_2^4) - (\omega^3 + 4\omega^2 + 3\omega + 2) (A_1^3 A_2 B_1^4) \\
    \phantom{ T_{1} = } - (5\omega^3 + 3\omega^2 + 3\omega + 4) (A_1^3 A_2 B_1^3 B_2) - (2\omega^3 + 3\omega^2 + 2\omega - 2) (A_1^3 A_2 B_1^2 B_2^2) \\
    \phantom{ T_{1} = } + (-2\omega^3 + \omega^2 + 1) (A_1^3 A_2 B_1 B_2^3) + (4\omega^2 - 2\omega + 3) (A_1^3 A_2 B_2^4) \\
   \phantom{ T_{1} = } + (-\omega^2 + 3\omega + 3) (A_1^2 A_2^2 B_1^4) + (2\omega^3 + 7\omega^2 + \omega) (A_1^2 A_2^2 B_1^3 B_2) \\
   \phantom{ T_{1} = } - (4\omega^3 + 6\omega^2 + 5\omega + 5) (A_1^2 A_2^2 B_1^2 B_2^2) - (3\omega^3 + 2) (A_1^2 A_2^2 B_1 B_2^3) \\
   \phantom{ T_{1} = } + (5\omega^3 + \omega + 4) (A_1^2 A_2^2 B_2^4) + (-2\omega^3 - 4\omega^2 + 1) (A_1 A_2^3 B_1^4) \\
   \phantom{ T_{1} = } + (\omega^3 + \omega^2 + \omega + 2) (A_1 A_2^3 B_1^3 B_2) + (\omega^3 + 3\omega + 1) (A_1 A_2^3 B_1^2 B_2^2) \\
   \phantom{ T_{1} = } - (7\omega^3 + 7\omega^2 + 4\omega + 7) (A_1 A_2^3 B_1 B_2^3) - (3\omega^3 + 2) (A_1 A_2^3 B_2^4) \\
   \phantom{ T_{1} = } + (-2\omega^3 + 2\omega^2 + 3\omega +2) (A_2^4 B_1^4) - (3\omega^3 + 1\omega^2 + 3\omega + 3) (A_2^4 B_1^3 B_2) \\
   \phantom{ T_{1} = } - (2\omega^3 + 3\omega + 5) (A_2^4 B_1^2 B_2^2) - (4\omega^3 + 5\omega^2 + 2\omega + 4) (A_2^4 B_1 B_2^3) \\
   \phantom{ T_{1} = } - (4\omega^3 + 6\omega^2 + 5\omega + 5) (A_2^4 B_2^4).
  $

\end{document}